\documentclass[10pt,a4paper]{article}
\makeatletter
\AtBeginDocument{%
  \@ifpackageloaded{hyperref}
  {\def\@doi#1{\href{https://doi.org/#1}
      {\ttfamily https://doi.org/#1}\egroup}}
  {\def\@doi#1{\ttfamily https://doi.org/#1\egroup}}
  \def\doi{\bgroup\catcode`\_=12\relax\@doi}}
\makeatother
\usepackage{authblk}
\usepackage[english]{babel}

\makeatletter
\renewcommand{\subsubsection}{\@startsection{subsubsection}{3}{\z@}%
                       {-12\p@ \@plus -4\p@ \@minus -4\p@}%
                       {-0.5em \@plus -0.22em \@minus -0.1em}%
                       {\normalfont\normalsize\bfseries\boldmath}}
\makeatother

\usepackage[T1]{fontenc}
\usepackage[colorlinks,bookmarksopen,bookmarksnumbered,allcolors=blue]{hyperref}
\usepackage{xcolor}

\usepackage{graphicx}
\usepackage{amssymb}
\usepackage{amsmath}
\usepackage{stmaryrd}
\usepackage{varwidth}
\usepackage{siunitx}

\usepackage{arydshln}
\usepackage{marvosym}
\usepackage{booktabs}
\usepackage{multirow}

\usepackage{amssymb,amsmath,amsthm}

\newtheorem{definition}{Definition}[section]
\newtheorem*{definition*}{Definition}
\newtheorem{theorem}{Theorem}[section]
\newtheorem{corollary}{Corollary}[theorem]
\newtheorem{lemma}[theorem]{Lemma}

\usepackage{url}
\usepackage{hyperref}
\usepackage{color}

\hypersetup{
  colorlinks   = true, 
  urlcolor     = blue, 
  linkcolor    = blue, 
  citecolor   = blue 
}

\renewcommand{\geq}{\geqslant}
\renewcommand{\leq}{\leqslant}

\newcommand{\fv}[1]{\ensuremath{\mathsf{fv}({#1})}}
\newcommand{\succs}[1]{\ensuremath{\mathord{\downarrow}{#1}}}

\newcommand{\red}{\mathbin{{\rightarrow}\!{\bullet}}}
\newcommand{\agg}{\mathbin{{\circ}\!{\rightarrow}}}

\newcommand{\maseq}[1]{\underline{#1}}

\newcommand{\comma}{\wedge}
\newcommand{\st}{\;.\;}

\newcommand{\Nat}{\mathbb{N}}
\newcommand{\Int}{\mathbb{Z}}

\newcommand{\eimplies}[1][E]{\mathop{\sqsubseteq_{#1}}}
\newcommand{\eequiv}[1][E]{\mathop{\equiv_{#1}\,}}

\newcommand{\FM}[1][A]{\mathrm{FM}_{{{#1}}}}

\newcommand{\reduc}{\equiv}

\newcommand*{\defeq}{\triangleq}
\makeatletter
\def \rightarrowfill{\m@th\mathord{\smash-}\mkern-6mu%
  \cleaders\hbox{$\mkern-2mu\mathord{\smash-}\mkern-2mu$}\hfill
  \mkern-6mu\mathord\rightarrow}
\makeatother
\makeatletter
\def \Rightarrowfill{\m@th\mathord{\smash-}\mkern-6mu%
  \cleaders\hbox{$\mkern-2mu\mathord{\smash-}\mkern-2mu$}\hfill
  \mkern-6mu\mathord\Rightarrow}
\makeatother
\makeatletter
\def \rightarrowfill{\m@th\mathord{\smash-}\mkern-6mu%
  \cleaders\hbox{$\mkern-2mu\mathord{\smash-}\mkern-2mu$}\hfill
  \mkern-6mu\mathord\rightarrow}
\makeatother
\makeatletter
\def \Rightarrowfill{\m@th\mathord{\smash=}\mkern-6mu%
  \cleaders\hbox{$\mkern-2mu\mathord{\smash=}\mkern-2mu$}\hfill
  \mkern-6mu\mathord\Rightarrow}
\makeatother
\makeatletter
\def \midrightarrowfill{\m@th\mathord{\smash{\raisebox{.2ex}{$\scriptscriptstyle\mid$}}\!\!\,-}\mkern-6mu%
  \cleaders\hbox{$\mkern-2mu\mathord{\smash-}\mkern-2mu$}\hfill
  \mkern-6mu\mathord\rightarrow}
\makeatother
\makeatletter
\def \midRightarrowfill{\m@th\mathord{\smash{\raisebox{.1ex}{$\scriptstyle\mid$}}\!\!\!=}\mkern-6mu%
  \cleaders\hbox{$\mkern-2mu\mathord{\smash=}\mkern-2mu$}\hfill
  \mkern-6mu\mathord\Rightarrow}
\makeatother
\newcommand{\overstackrel}[2]{\mathrel{\mathop{#1}\limits^{#2}}}
\newcommand{\trans}[1]{\mathbin{\smash[t]{\overstackrel{\rightarrowfill}{\ #1\ }}}}

\newcommand{\tfg}[2][{}]{\mathrm{T}_{#1}({#2})}

\newcommand{\tfgA}{A^{\circ}}
\newcommand{\tfgR}{R^{\bullet}}

\newcommand{\C}{\mathit{C}}
\newcommand{\B}{\mathit{B}}
\newcommand{\HLF}[2][{X}]{\mathrm{HLF}_{#1}({#2})}

\newcommand{\tool}[1]{\textsf{#1}}
\newcommand{\reduce}{\tool{Reduce}}
\newcommand{\tina}{\tool{Tina}}
\newcommand{\lola}{\tool{LoLA}}
\newcommand{\its}{\tool{ITS}}
\newcommand{\tapaal}{\tool{TAPAAL}}

\providecommand{\keywords}[1]
{\small\textbf{\textit{Keywords---}} #1}

\title{Project and Conquer: Fast Quantifier Elimination for Checking Petri Net Reachability}
\author[1]{Nicolas Amat}
\author[1]{Silvano {Dal Zilio}} 
\author[1]{Didier {Le Botlan}} 
\affil[1]{LAAS-CNRS, Universit\'{e} de Toulouse, CNRS, Toulouse, France}
\date{}
\begin{document}
\maketitle
\sloppy

\begin{abstract}
  We propose a method for checking generalized reachability properties in Petri
  nets that takes advantage of structural reductions and that can be used,
  transparently, as a pre-processing step of existing model-checkers. Our
  approach is based on a new procedure that can project a property, about
  an initial Petri net, into an equivalent formula that only refers to the
  reduced version of this net.
  Our projection is defined as a variable elimination procedure for linear
  integer arithmetic tailored to the specific kind of constraints we handle.
  It has linear complexity, is guaranteed to return a sound property, and makes
  use of a simple condition to detect when the result is exact. Experimental
  results show that our approach works well in practice and that it can be
  useful even when there is only a limited amount of reductions.
  \vskip 0.3em
  \noindent\keywords{Petri nets; Quantifier elimination; Reachability problems.}
\end{abstract}
%
%

\section{Introduction}

We describe a method to accelerate the verification of reachability properties
in Petri nets by taking advantage of {structural
reductions}~\cite{berthelot_transformations_1987}. We focus on the verification
of generalized properties, that can be expressed using a Boolean combination of
linear constraints between places, such as $(2\, p_0 + p_1 = 5) \land (p_1 \geq
p_2)$ for example. This class of formulas corresponds to the reachability
queries used in the Model Checking Contest (MCC)~\cite{mcc2019}, a competition
of Petri net verification tools that we use as a benchmark.

In essence, net reductions are a class of transformations that can simplify an
initial net, $(N_1, m_1)$, into another, residual net $(N_2, m_2)$, while
preserving a given class of properties. This technique has become a conventional
optimization integrated into several model-checking
tools~\cite{berthomieu2018petri,bonneland_stubborn_2019,thierry2021symbolic}. A
contribution of our paper is a procedure to transform a property $F_1$,
about the net $N_1$, into a property $F_2$ about the reduced net $N_2$, while
preserving the verdict. We have implemented this procedure into a new tool,
called \tool{Octant}~\cite{octant}, that can act as a pre-processor allowing any
model-checker to transparently benefit from our optimization. Something that was
not possible in our previous works. In practice, it means that we can use our
approach as a front-end to accelerate any model-checking tool that supports
generalized reachability properties, without modifying them. 

Our approach relies on a notion, called \emph{polyhedral
reduction}~\cite{pn2021,fi2022,spin2021,berthomieu_counting_2019}, that
describes a linear dependence relation, $E$, between the reachable markings of a
net and those of its reduced version. This equivalence, denoted $(N_1, m_1)
\eequiv (N_2,m_2)$, preserves enough information in $E$ so that we can rebuild
the state space of $N_1$ knowing only the one of $N_2$. An interesting
application of this relation is the following \emph{reachability conservation
theorem}~\cite{pn2021}: assume we have $(N_1, m_1) \eequiv (N_2,m_2)$, then
property $F$ is reachable in $N_1$ if and only if $E \land F$ is reachable in
$N_2$.
This property is interesting since it means that we can apply more aggressive
reduction techniques than, say,
\emph{slicing}~\cite{rakow2012safety,llorens2017integrated,khan2018survey},
\emph{cone of influence}~\cite{clarke2018model}, or other
methods~\cite{ganty2008many,kang2021abstraction} that seek to remove or gather
together places that are not relevant to the property we want to check. We do
not share this restriction in our approach, since we reduce nets beforehand and
can therefore reduce places that occur in the initial property. We could argue
that approaches similar to slicing only simplify a model with respect to a
formula, whereas, with our method, we simplify the model as much as possible and
then simplify formulas as needed. This is more efficient when we need to check
several properties on the same model and, in any case, nothing prevents us from
applying slicing techniques on the result of our projection.

However, there is a complication, arising from the fact that the formula $E
\land F$ may include variables (places) that no longer occur in the reduced net
$N_2$, and therefore act as existentially quantified variables. This can
complicate some symbolic verification techniques, such as
$k$-induction~\cite{sheeran_checking_2000}, and impede the use of explicit,
enumerative approaches. Indeed, in the latter case, it means that we need to
solve an integer linear problem for each new state, instead of just evaluating a
closed formula. To overcome this problem, we propose a new method for projecting
the formula $E \land F$ into an equivalent one, $F'$, that only refers to the
places of $N_2$.
We define our projection as a procedure for quantifier elimination in
Presburger Arithmetic (PA) that is tailored to the specific kind of constraints
we handle in $E$. Whereas quantifier elimination has an exponential
complexity in general for existential formulas, our construction has linear complexity and can only
decrease the size of a formula. It also always terminates and returns a result
that is guaranteed to be sound; meaning it under-approximates the set of
reachable models and, therefore, a witness of $F'$ in $N_2$ necessarily
corresponds to a witness of $F$ in $N_1$.
Additionally, our approach includes a simple condition on $F$ that is enough to
detect when our result is exact, meaning that if $F'$ is unreachable in $N_2$,
then $F$ is unreachable in $N_1$. We show in Sect.~\ref{sec:performance} that
our projection is complete for $80\%$ of the formulas used in the MCC.


\subsubsection*{Outline and Contributions.}
We start by giving some technical background about Petri nets and the notion of
polyhedral abstraction in Sect.~\ref{sec:fundations}, then describe how to use
this equivalence to accelerate the verification of reachability properties
(Th.~\ref{th:exploration} in Sect.~\ref{sec:combining}). We also use this
section to motivate our need to find methods to eliminate (or project) variables
in a linear integer system. We define our fast projection method in
Sect.~\ref{sec:tfg}, which is based on a dedicated graph structure, called Token
Flow Graph (TFG), capturing the particular form of constraints occurring with
polyhedral reductions. We prove the correctness of this method in
Sect.~\ref{sec:procedure}. Our method has been implemented, and we report on the
results of our experiments in Sect.~\ref{sec:performance}. We give quantitative
evidence about several natural questions raised by our approach. We start by
proving the effectiveness of our optimization on both $k$-induction and random
walk. Then, we show that our method can be transparently added as a
preprocessing step to off-the-shelf verification tools. This is achieved by
testing our approach with the three best-performing tools that participated in
the reachability category of the MCC---\tool{ITS-Tools}~\cite{its_tools} (or
\its\ for short); \lola~\cite{lola}; and \tapaal~\cite{tapaal}---which are
already optimized for the type of models and formulas used in our benchmark.
Our results show that reductions are effective on a large set of queries and
that their benefits do not overlap with other existing optimizations, an
observation that was already made in~\cite{fi2022,bonneland_stubborn_2019}. We
also prove that our procedure often computes an exact projection and compares
favorably well with the variable elimination methods implemented in
\tool{isl}~\cite{verdoolaege2010isl} and
\tool{Redlog}~\cite{dolzmann1997redlog}. This supports our claim that we are
able to solve non-trivial quantifier elimination problems.


\section{Petri Nets and Polyhedral Abstraction}
\label{sec:fundations}

Most of our results involve non-negative integer solutions to constraints
expressed in Presburger Arithmetic, the first-order theory of the integers with
addition~\cite{haase2018survival}.
We focus on the quantifier-free fragment of PA, meaning Boolean combinations
(using $\land$, $\lor$ and $\neg$) of atomic propositions of the form $\alpha
\sim \beta$, where $\sim$ is one of $=, \leq$ or $\geq$, and $\alpha, \beta$ are
linear expressions with coefficients in $\Int$. Without loss of generality, we
can consider only formulas in disjunctive normal form (DNF), with \emph{linear
predicates} of the form $(\sum k_i\,x_i) + b \geqslant 0$. We deliberately do
not add a divisibility operator $k \mid \alpha$, which requires that $k$ evenly
divides $\alpha$, since it can already be expressed with linear predicates,
though at the cost of an extra existentially quantified variable.
This fragment corresponds to the set of reachability formulas supported by many
model-checkers for Petri nets, such
as~\cite{30YearsOfGreatSPN,tina2004,tapaal,its_tools,lola}.

We use $\Nat^V$ to denote the space of mappings over $V = \{x_1, \dots, x_n\}$,
meaning total mappings from $V$ to $\Nat$. We say that a mapping $m$ in $\Nat^V$
is a \emph{model} of a quantifier-free formula $F$ if the variables of $F$,
denoted $\fv{F}$, are included in $V$ and the closed formula $F\{m\}$ (the
substitution of $m$ in $F$) is true. We denote this relation $m \models F$.
\begin{equation}
  \label{eq:1}
  F\{m\} \ \defeq\ F\{x_1 \leftarrow m(x_1)\}\dots\{x_n \leftarrow
m(x_n)\}
\end{equation}
We say that a Presburger formula is \emph{consistent} when it has at least one
model. We can use this notion to extend the definition of models in the case
where $F$ is over-specified; i.e. it has a larger support than $m$. If $m' \in
\Nat^{U}$ with $U \subseteq \fv{F}$, we write $m' \models F$ when $F\{m'\}$ is
consistent. 

\subsubsection*{Petri Nets and Reachability Formulas.}
A \textit{Petri net} $N$ is a tuple $(P, T, \mathrm{Pre},\allowbreak
\mathrm{Post})$ where $P = \{p_1, \dots, p_n\}$ is an ordered set of places, $T
= \{t_1, \dots, t_k\}$ is a finite set of transitions (disjoint from $P$), and
$\mathrm{Pre} : T \rightarrow \Nat^P$ and $\mathrm{Post} : T \rightarrow \Nat^P$
are the pre- and post-condition functions (also called the flow functions of
$N$).
A state $m$ of a net, also called a \emph{marking}, is a mapping of $\Nat^P$. A
marked net $(N, m_0)$ is a pair composed of a net and its initial marking $m_0$.

We extend the comparison $(=,\geqslant)$ and arithmetic operations $(-,+)$ to
their point-wise equivalent. With our notations, a transition $t \in T$ is said
\textit{enabled} at marking $m$ when $m \geq \mathrm{Pre}(t)$. A marking $m'$ is
reachable from a marking $m$ by firing transition $t$, denoted $m \trans{t} m'$,
if: (1) transition $t$ is enabled at $m$; and (2) $m' = m - \mathrm{Pre}(t) +
\mathrm{Post}(t)$. When the identity of the transition is unimportant, we simply
write this relation $m \trans{} m'$. More generally, a marking $m'$ is reachable
from $m$ in $N$, denoted $m \trans{}^\star m'$ if there is a (possibly empty)
sequence of transitions such that $m \trans{} \dots \trans{} m'$. We denote
$R(N, m_0)$ the set of markings reachable from $m_0$ in $N$.

We are interested in the verification of properties over the reachable markings
of a marked net $(N, m_0)$, with a set of places $P$. Given a formula $F$ with
variables in $P$, we say that $F$ is reachable if there exists at least one
reachable marking, $m \in R(N, m_0)$, such that $m \models F$. We call such
marking a \emph{witness} of $F$. Likewise, $F$ is said to be an \emph{invariant}
when all the reachable markings of $(N, m_0)$ are models of $F$.
This corresponds to the two classes of queries found in our benchmark:
$\mathrm{EF} \, F$, which is true only if $F$ is reachable; and $\mathrm{AG} \,
F$, which is true when $F$ is an invariant, with the classic relationship that
$\mathrm{AG}\, F \equiv \neg \, (\mathrm{EF}\, \neg F)$. Examples of
properties we can express in this way include: checking if some transition can
possibly be enabled, checking if there is a deadlock, checking whether some
linear invariant between places is always true, etc.

\begin{figure}[tb]
  \centering
  \includegraphics[width=0.9\textwidth]{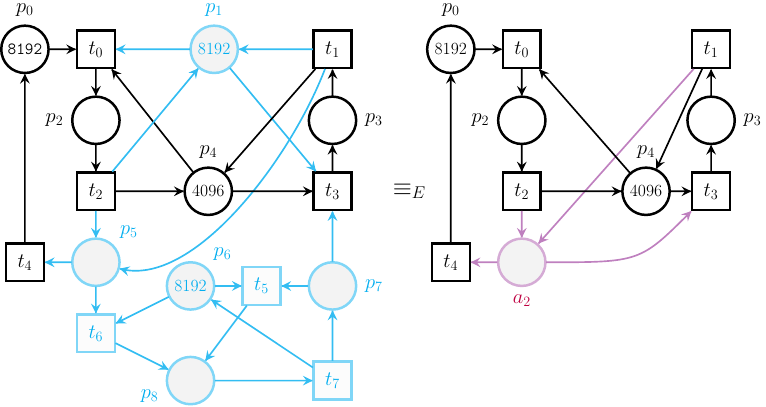}
  \caption{An example of Petri net, $M_1$ (left), and one of its polyhedral
  abstractions, $M_2$ (right), with $E \defeq (p_1 = p_4 + 4096) \comma (p_6 =
  p_0 + p_2 + p_3 + p_5 + p_7) \comma (a_1 = p_7 + p_8) \comma (a_2 = a_1 +
  p_5)$. Colors are used to emphasize places that are either removed or
  added.\label{fig:stahl}}
\end{figure}

We use a standard graphical notation for nets where places are depicted with
circles and transitions with squares. We give an example in
Fig.~\ref{fig:stahl}, where net $M_1$ depicts the SmallOperatingSystem model,
borrowed from the MCC benchmark~\cite{kordon_small_2015}. This net abstracts the
lifecycle of a task in a simplified operating system handling several memory
segments (place $p_0$), disk controller units ($p_4$), and cores ($p_6$). The
initial marking of the net gives the number of resources available (e.g., there
are $8\,192$ available memory segments in our example).

We chose this model since it is one of the few examples in our benchmark that
fits on one page. This is not to say that the example is simple. Net $M_1$ has
about $10^{17}$ reachable states, which means that it is out of reach of
enumerative methods, and only one symbolic tool in the MCC is able to generate
its whole state space\footnote{It is \tool{Tedd}, part of the Tina toolbox,
which also uses polyhedral reductions.}~\cite{mcc:2022}. For comparison, the
reduced net $M_2$ has about $10^{10}$ states.

\subsubsection*{Polyhedral Abstraction.}
We recently defined an equivalence relation that describes linear dependencies
between the markings of two different nets, $N_1$ and $N_2$~\cite{fi2022}. In
the following, we reserve $F$ for formulas about a single net and use $E$ to
refer to relations. Assume $m$ is a mapping of $\Nat^V$. We can associate $m$ to
the linear predicate $\maseq{m}$, which is a formula with a unique model $m$.
\begin{equation}
  \label{eq:2}
  \maseq{m} \ \defeq\ \bigwedge \{ x = m(x) \mid x \in V \}
\end{equation}
By extension, we say that $m$ is a (partial) solution of $E$ if the system $E
\comma \maseq{m}$ is consistent. In some sense, we use $\maseq{m}$ as a
substitution, since the formulas $E\{m\}$ and $E \land \maseq{m}$ have the same
models. Given two mappings $m_1 \in \Nat^{V_1}$ and $m_2 \in \Nat^{V_2}$, we say
that $m_1$ and $m_2$ are \emph{compatible} when they have equal values on their
shared domain: $m_1(x) = m_2(x)$ for all $x$ in $V_1 \cap V_2$. This is a
necessary and sufficient condition for the system $\maseq{m_1} \comma
\maseq{m_2}$ to be consistent. Finally, we say that $m_1$ and $m_2$ are related
up-to $E$, denoted $m_1 \eequiv m_2$, when $E \comma \maseq{m_1} \comma
\maseq{m_2}$ is consistent.
\begin{equation}
  \label{eq:3}
  m_1 \eequiv m_2  \quad \Leftrightarrow \quad \exists m \in \Nat^V
  \st m \models E \comma \maseq{m_1} \comma \maseq{m_2}
\end{equation}
This relation defines an equivalence between markings of two different nets
($\eequiv \subseteq \Nat^{P_1} \times \Nat^{P_2}$) and, by extension, can be
used to define an equivalence between nets themselves, that we call
\emph{polyhedral equivalence}.

\begin{definition}[$E$-equivalence]
  \label{def:eequiv}
  We say that $(N_1, m_1)$ is $E$-equivalent to\linebreak $(N_2, m_2)$, denoted
  $(N_1, m_1) \reduc_E (N_2, m_2)$, if and only if:
  \begin{description}
 \item [(A1)] $E \comma \maseq{m}$ is consistent for all markings $m$     in
 $R(N_1, m_1)$ or $R(N_2, m_2)$;

  \item[(A2)] initial markings are \emph{compatible}: $m_1 \eequiv m_2$;

  \item[(A3)] assume $m'_1, m'_2$ are markings of $N_1,N_2$, such that $m'_1 \eequiv m'_2$, then 
  $m_1'$ is reachable iff $m_2'$ is reachable: $m'_1 \in R(N_1,m_1) \iff m_2' \in R(N_2, m_2)$.
  \end{description}
\end{definition}

By definition, given the equivalence $(N_1, m_1) \reduc_E (N_2, m_2)$, every
marking $m'_2$ reachable in $N_2$ can be associated to a subset of markings in
$N_1$, defined from the solutions to $E\comma \maseq{m'_2}$ (by condition (A1) and
(A3)). In practice, this gives a partition of the reachable markings of $(N_1,
m_1)$ into ``convex sets''---hence the name polyhedral abstraction---each
associated with a reachable marking in $N_2$. This approach is particularly
useful when the state space of $N_2$ is very small compared to the one of $N_1$.

We can prove that the two marked nets in our running example satisfy $M_1
\eequiv M_2$, for the relation $E$ defined in Fig.~\ref{fig:stahl}. Net $M_2$ is
obtained automatically from $M_1$ by applying a set of reduction rules,
iteratively, and in a compositional way. This process relies on the reduction
system defined in~\cite{berthomieu_counting_2019,fi2022}. As a result, we manage
to remove five places: $p_1, p_5, p_6, p_7, p_8$, and only add a new one, $a_2$.
The ``reduction system'' ($E$) also contains an extra variable, $a_1$, that does
not occur in any of the nets. It corresponds to a place that was introduced and
then removed in different reduction steps.

Polyhedral abstractions are not necessarily derived from reductions, but
reductions provide a way to automatically find interesting instances of
abstractions. Also, the equation systems obtained using structural reductions
exhibit a specific structure, that we exploit in Sect.~\ref{sec:tfg}.

\section{Combining Polyhedral Abstraction with Reachability}
\label{sec:combining}

We can define a counterpart to our notion of polyhedral abstraction which
relates to reachability formulas. We show that this equivalence can be used to
speed up the verification of properties by checking formulas on a reduced net
instead of the initial one (see Th.~\ref{th:exploration} and its corollary,
below). In the following, we assume that we have two marked nets such that
$(N_1, m_1) \reduc_E (N_2, m_2)$. Our goal is to define a relation $F_1 \eequiv
F_2$, between reachability formulas, such that $F_1$ and $F_2$ have the same
truth values on equivalent models, with respect to $E$.

\begin{definition}[Equivalence between formulas]
  \label{def:feequiv}
  Assume $F_1, F_2$ are reachable formulas with respective sets of variables,
  $V_1$ and $V_2$, in the support of $E$. We say that formula $F_2$ implies
  $F_1$ up-to $E$, denoted $F_2 \eimplies F_1$, if for every marking $m_2' \in
  \Nat^{V_2}$ such that $m'_2 \models E \land F_2$ there exists at least one
  marking $m_1' \in \Nat^{V_1}$ such that $m_1' \eequiv m_2'$ and $m'_1 \models
  E \land F_1$.
  \begin{equation} \label{eq:4} F_2 \eimplies F_1 \quad \text{iff}\quad  
   \forall m'_2 . \left ( m'_2 \models E \land F_2 \right ) \, \Rightarrow\,
   \exists m'_1 . \left ( m'_1 \eequiv m'_2  \land m'_1 \models E \land F_1
   \right ) \end{equation}
  We say that $F_1$ and $F_2$ are equivalent, denoted $F_1 \reduc_E F_2$, when
  both $F_1 \eimplies F_2$ and $F_2 \eimplies F_1$.
\end{definition}

This notion is interesting when $F_1, F_2$ are reachability formulas on the nets
$N_1$, respectively $N_2$. Indeed, we prove that when $F_2 \eimplies F_1$, it is
enough to find a witness of $F_2$ in $N_2$ to prove that $F_1$ is reachable in
$N_1$.

\begin{theorem}[Finding Witnesses]
  \label{th:exploration}
  Assume $(N_1, m_1) \reduc_E (N_2, m_2)$ and $F_2 \eimplies F_1$, and take a
  marking $m'_2$ reachable in $(N_2, m_2)$ such that $m'_2 \models F_2$. Then
  there exists $m'_1 \in R(N_1, m_1)$ such that $m'_1 \eequiv m'_2$ and $m'_1
  \models F_1$.
\end{theorem}
\begin{proof}
  Assume we have $m'_2$ reachable in $N_2$ such that $m'_2 \models F_2$.
  By property (A1) of $E$-equivalence (Def.~\ref{def:eequiv}), formula $E \land
  \maseq{m'_2}$ is consistent, which gives $m'_2 \models E \wedge F_2$.
  By definition of the $E$-implication $F_2 \eimplies F_1$, we get a marking
  $m'_1$ such that $m'_1 \models F_1$ and $m'_1 \eequiv m'_2$. We conclude that
  $m'_1$ is reachable in $N_1$ thanks to property~(A3).
  \qed
\end{proof}
Hence, when $F_2 \eimplies F_1$ holds, $F_2$ reachable in $N_2$ implies that
$F_1$ is reachable in $N_1$. We can derive stronger results when $F_1$ and $F_2$
are equivalent. 
\begin{corollary}
  Assume $(N_1, m_1) \reduc_E (N_2, m_2)$ and $F_1 \equiv_E F_2$, with $\fv{F_i}
  \subseteq P_i$ for all $i \in 1..2$, then: \emph{(CEX)} property $F_1$ is
  reachable in $N_1$ if and only if $F_2$ is reachable in $N_2$ ; and
  \emph{(INV)} $F_1$ is an invariant on $N_1$ if and only if $F_2$ is an
  invariant on $N_2$. 
\end{corollary}

Theorem~\ref{th:exploration} means that we can check the reachability (or
invariance) of a formula on the net $N_1$ by checking instead the reachability
of another formula ($F_2$) on $N_2$. But it does not indicate how to compute a
good candidate for $F_2$. By Definition~\ref{def:feequiv}, a natural choice is
to select $F_2 \defeq E \land F_1$. 
We can actually do a bit better. It is enough to choose a formula $F_2$ that has
the same (integer points) solution as $E \land F_1$ over the places of $N_2$.
More formally, let ${A} \defeq \fv{E} \setminus P_2$ be the set of ``additional
variables'' from $E$; variables occurring in $E$ which are not places of the
reduced net $N_2$. Then if $F_2$ has the same integer solutions over
$\Nat^{P_2}$ than the Presburger formula $\exists A . \, (E \land F_1)$, we have
$F_1 \eequiv F_2$. We say in this case that $F_2$ is the projection of $E \land
F_1$ on the set $P_2$, by eliminating the variables in $A$. 

In the next section, we show how to compute a candidate projection formula
without resorting to a classical, complete variable elimination procedure on $E
\land F_1$. This eliminates a potential source of complexity blow-up.

We can use Fourier-Motzkin elimination (FM) as a point of reference. Given a
system of linear inequalities $S$, with variables in $V$, we denote $\FM(S)$ the
system obtained by FM elimination of variables in $A$ from $S$. (We do not
describe the construction of $\FM(S)$ here, since there exists many good
references~\cite{imbert1993fourier,monniaux2010quantifier} on the subject.)
Borrowing an intuition popularized by Pugh in its Omega test~\cite{omega}, we
can define two distinct notions of ``shadows'' cast by the projection of $S$. On
the one hand, we have the \emph{real shadow}, relative to $A$, which are the
integer points (in $\Nat^{V \setminus A}$) solutions of $\FM(S)$. On the other
hand, the \emph{integer shadow} of $S$ is the set of markings $m'$ with an
integer point antecedent in $S$. We need the latter to check a query on $N_1$. A
main source of complexity is that the (real) shadow is only exact on rational
points and may contain strictly more models than the integer shadow. Moreover,
while the real shadow of a convex region will always be convex, it may not be
the case with the integer shadow. Like with the real shadow, the set of
equations computed with our fast projection will always be convex. Unlike FM,
our procedure will compute an under-approximation of the integer shadow, not an
over-approximation. Also, we never rearrange or create more inequalities than
the one contained in $S$; but instead rely on variable substitution.

We illustrate the concepts introduced in this section on our running example,
with the reduction system from Fig.~\ref{fig:stahl}. With our notations, we try
to eliminate variables in $A \defeq \{a_1, p_1, p_5, p_6, p_7, p_8\}$ and keep
only those in $P_2 \defeq \{a_2, p_0, p_2, p_3, p_4\}$.

Take the formula $G_1 \defeq (p_5 + p_6 \leqslant p_8)$. Using substitutions
from constraints in $E$, namely the fact that $(a_2 = p_7 + p_8 + p_5)$ and
$(p_6 = p_0 + p_2 + p_3 + p_5 + p_7)$, we can remove occurrences of $a_1, p_1,
p_6, p_8$ from $E \land G_1$, leaving the resulting equation $(3\,p_5 + 2\,p_7 +
p_0 + p_2 + p_3 \leqslant a_2) \land (a_2 = p_5 + p_7 + p_8)$, that still refers
to $p_5$ and $p_7$.
We observe that non-trivial coefficients (like $3\,p_5$) can naturally occur
during this process, even though all the coefficients are either $1$ or $-1$ in
the initial constraints.
We can remove the remaining variables to obtain an exact projection of $G_1$
using our fast projection method, described below.
The result is the formula $G_2 \defeq (p_0 + p_2 + p_3 \leqslant a_2)$. 

Another example is $H_1 \defeq (p_6 = p_8)$. We can prove that the integer
shadow of $E \land H_1$, after projecting the variables in $A$, are the
solutions to the PA formula $(a_2 - p_0 - p_2 - p_3 \equiv 0 \bmod 2) \land (a_2
\geqslant p_0 + p_2 + p_3)$. This set is not convex, since $(a_2 = 0 \land p_0 =
p_1 = p_2 = p_3 = 0)$ and $(a_2 = 2 \land p_0 = p_1 = p_2 = p_3 = 0)$ are in the
integer shadow, but not $(a_2 = 1 \land p_0 = p_1 = p_2 = p_3 = 0)$ for
instance. Our fast projection method will compute the formula $H_2 \defeq (a_2 +
p_0 + p_2 + p_3 = 0)$ and flag it as an under-approximation.


\section{Projecting Formulas Using Token Flow Graphs}
\label{sec:tfg}
We describe a formula projection procedure that is tailored to the specific kind
of constraints occurring in polyhedral reductions.
The example in Fig.~\ref{fig:stahl} is representative of the ``shape'' of
reduction systems: it mostly contains equalities of the form $x = \sum x_i$,
over a sparse set of variables, but may also include some inequalities; and it
can have a very large number of literals (often proportional to the size of the
initial net). Another interesting feature is the absence of cyclic dependencies,
which underlines a hierarchical relationship between variables.

We can find a more precise and formal description of these constraints
in~\cite{spin2021}, with the definition of a \textit{Token Flow Graph} (TFG).
Basically, a TFG for a reduction system $E$ is a directed acyclic graph (DAG)
with one vertex for each variable occurring in $E$.
We consider two kinds of arcs, redundancy ($\red$) and agglomeration ($\agg$),
that correspond to two main classes of reduction rules.

Arcs for \emph{redundancy equations}, $q \red p$, correspond to equations of the
form $p = q + r + \dots$, expressing that the marking of place $p$ can be
reconstructed from the marking of $q, r, \dots$ In this case, we say that place
$p$ is \emph{removed} by arc $q \red p$, because the marking of $q$ may
influence the marking of $p$, but not necessarily the other way round.

Arcs for \emph{agglomeration equations}, $a \agg p$, represent equations of the
form $a = p + q + \dots$, generated when we agglomerate several places into a
new one. In this case, we expect that if we can reach a marking with $k$ tokens
in $a$, then we can certainly reach a marking with $k_1$ tokens in $p$, and
$k_2$ tokens in $q$, $\dots$ such that $k = k_1 + k_2 + \dots$. Hence, the
possible markings of $p$ and $q$ can be reconstructed from the markings of $a$.
In this case, it is $p, q, \dots$ which are removed. We also say that node $a$
is \emph{inserted}; it does not exist in $N_1$ but may appear as a new place in
$N_2$ unless it is removed by a subsequent reduction. We can have more than two
places in an agglomeration, see the rules in~\cite{berthomieu_counting_2019}.

\begin{figure}[tb]
  \centering

  \begin{minipage}{0.30\textwidth}
    \centering
    {
      \begin{varwidth}{\linewidth}
        \small
\begin{verbatim}
# R |- p1 = p4 + 4096
# R |- p6 = p0 + p2 + p3 + p5 + p7
# A |- a1 = p7 + p8
# A |- a2 = a1 + p5
\end{verbatim}
      \end{varwidth}}
  \end{minipage}\hfill
  \begin{minipage}{0.60\textwidth}
    \centering
    \includegraphics[width=\textwidth]{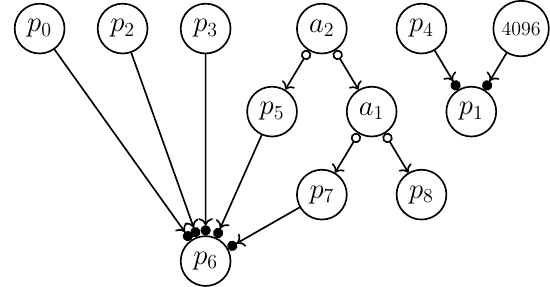}
  \end{minipage}

  \caption{Equations from our example in
    Fig.~\ref{fig:stahl} and the associated TFG.}
  \label{fig:HuConstruction_TFG}
\end{figure}

A TFG can also include nodes for \emph{constants}, used to express invariant
statements on the markings of the form $p + q = k$. To this end, we assume that
we have a family of disjoint sets $K(n)$ for each $n$ in $\Nat$, such that the
``valuation'' of a node $v \in K(n)$ is always $n$. We use $K$ to denote the set
of all constants.

\begin{definition}[Token Flow Graph]
  \label{def:tfg}
  A TFG with set of places $P$ is a directed graph $(P, S, \tfgR, \tfgA)$ such
  that:
  \begin{itemize}
    \item $V = P \cup S$ is a set of vertices (or nodes) with $S \subset K$ a
    finite set of constants,
    \item $\tfgR \in V\times V$ is a set of \emph{redundancy arcs}, $v \red v'$,
    \item $\tfgA \in V\times V$ is a set of \emph{agglomeration arcs}, $v \agg
    v'$, disjoint from $R$.
  \end{itemize}
\end{definition}

The main source of complexity in this approach arises from the need to manage
interdependencies between $\tfgA$ and $\tfgR$ arcs, that is situations where
redundancies and agglomerations are combined. This is not something that can be
easily achieved by looking only at the equations in $E$ and thus motivates the
need for a specific data structure.

We define several notations that will be useful in the following. We use the
notation $v \rightarrow v'$ when we have $(v \red v')$ or $(v \agg v')$. We say
that a node $v$ is a \textit{root} if it is not the target of an arc.
A sequence of nodes $(v_1, \dots, v_n)$ in $V^n$ is a \textit{path} if for all
$1 \leqslant i < n$ we have $v_i \rightarrow v_{i+1}$. We use the notation $v
\rightarrow^\star v'$ when there is a path from $v$ to $v'$ in the graph, or
when $v = v'$.
We write $v \agg X$ when $X$ is the largest subset $\{v_1, \dots, v_k\}$ of $V$
such that $X \neq \emptyset$ and $v \agg v_i$ for all $i \in 1..k$. And
similarly with reductions, $X \red v$.
Finally, the notation $\succs{v}$ denotes the set of successors of $v$, that is:
$\succs{v} \defeq \{ v' \in V \setminus \{v\} \,\mid\, v \rightarrow^\star v'
\}$.
We extend it to a set of variables $X$ with $\succs{X} = \bigcup_{x \in X}
\succs{x}$.

We display an example of TFG in Fig.~\ref{fig:HuConstruction_TFG} (right), which
corresponds to the reduction equations generated on our running example, where
annotations \texttt{R} and \texttt{A} indicate if an equation is a redundancy or
an agglomeration. TFGs were initially defined in~\cite{spin2021,sttt22} to
efficiently compute the set of concurrent places in a net, that is all pairs of
places that can be marked simultaneously in some reachable marking. We reuse
this concept here to project reachability formulas.

\subsubsection*{High Literal Factor.}

The projection procedure, described next, applies to \emph{cubes} only, meaning
a conjunction of literals $\bigwedge_{i \in 1..n} \alpha_i$.
Given a formula $F_1$, assumed in DNF, we can apply the projection procedure to
each of its cubes, separately. Then the projection of $F_1$ is the disjunction
of the projected cubes.
We assume from now on that $F_1$ is a cube formula.

We assume that every literal is in \emph{normal form}, $\alpha_i \defeq
(\sum_{p_j \in \B} k^i_j \, p_j) + b_i \geqslant 0$, where the $k^i_j$'s and
$b_i$ are in $\Int$. In the following, we denote $\alpha_i(q)$ the coefficient
associated with variable $q$ in $\alpha_i$. We also use $\max_X \alpha_i$ and
$\min_X \alpha_i$ for the maximal (resp. minimal) coefficient associated with
variables in $X \subseteq \B$.
\[
  \alpha_i \ =\ \sum_{p \in \B} \alpha_i(p) \, p + b_i
  \qquad \text{and} \qquad
  {\textstyle \max_X}\, \alpha_i \ =\ \max\ \left \{ \alpha_i(p) \mid p \in X \right \}
\]

We define the \emph{Highest Literal Factor} (HLF) of a set of variables $X$ with
respect to a set of normalized literals $(\alpha_i)_{i \in I}$. In the simplest
case, the HLF of $X$ with respect to a single literal, $\alpha$, is the subset
of variables in $X$ with the highest coefficients in $\alpha$. Then, the HLF of
$X$ with respect to a set of literals is the (possibly empty) intersection of
the HLFs of $X$ with respect to each literal.
When non-empty, it means that at least one variable in $X$ always has the
highest coefficient, and we say then that the whole set $X$ is \emph{polarized}
with respect to the literals~$(\alpha_i)$.
\begin{equation*}
  \begin{split}
    \HLF{\alpha_i} & =\ \left \{ p \in X \mid \alpha_i(p) =
      {\textstyle \max_X}\, \alpha_i \right \} \\
    \HLF{\alpha_i}_{i \in I} & =\ \bigcap_{i \in I}
    \HLF{\alpha_i}
  \end{split}
\end{equation*}

\begin{definition}[Polarized Set of Constraints]
  A set of variables $X \subseteq \fv{\C}$ is said {polarized} with respect to a
  set of normalized literals $\C$ when\linebreak $\HLF{\C} \neq \emptyset$.
\end{definition}

We prove in the next section that our procedure is exact when the variables we
eliminate are polarized. While this condition seems very restrictive, we observe
that it is often true with the queries used in our experiments.

\subsubsection*{Example.}

Let us illustrate our approach with two examples. Assume we want to eliminate an
agglomeration $a \agg \{q, r\}$, meaning that we have the condition $a = q + r$
and that both $q$ and $r$ must disappear. We consider two examples of systems,
each with only two literals, and with free variables $\{p, q\}$. 
\begin{center}
  \includegraphics{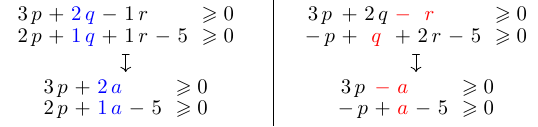}
\end{center}

In the left example, the set $\{q, r\}$ is polarized with respect to the initial
system (top), with the highest literal factor being $q$. So we replace $q$ with
$a$ in both literals and eliminate $r$. Uninvolved variables (the singleton
$\{p\}$ in this case) are left unchanged. We can prove that both systems, before
and after substitution, are equivalent. For instance, every solution in the
resulting system can be associated with a solution of the initial one by taking
$q = a$ and $r = 0$.

The initial system on the right (top) is non-polarized: the HLF relative to
$\{q, r\}$ is $\{q\}$ for the first literal ($+ 2\, q$ versus $- r$) and $\{r\}$
in the second ($+ q$ versus $2\, r$). So we substitute $a$ to the variable with
the lowest literal factor, in each literal, and remove the other variable ($r$
in the first literal and $q$ in the second). This is sound because we take into
account the worst case in each case. But this is not complete, because we may be
too pessimistic. For instance, the resulting system has no solution for $p = 2$;
because it entails $a \leq 6$ and $a \geq 7$. But $p = 2, q = 3, r = 2$ is a
model of the initial system.

Next, we give a formal description of our projection procedure as a sequence of
formula rewriting steps and prove that the result is exact (we have $F_2 \eequiv
F_1$) when all the reduction steps corresponding to an agglomeration are on
polarized variables.


\section{Formal Procedure and Proof of Correctness}
\label{sec:procedure}

In all the results of this section, we assume that $N_1$ and $N_2$ are two nets,
with respective sets of places $P_1, P_2$ and initial markings $m_1, m_2$, such
that $(N_1, m_1) \reduc_E (N_2, m_2)$. Given a formula $F_1$ with support on
$N_1$, we describe a procedure to project formula $E \land F_1$ into a new
formula, $F_2$, with support on $N_2$. Our projection will always lead to a
sound formula, meaning $F_2 \eimplies F_1$. It is also able in many cases (see
some statistics in Sect.~\ref{sec:performance}) to result in an exact formula,
such that $F_2 \eequiv F_1$.

\subsubsection*{Constraints on TFGs.}
To ensure that a TFG preserves the semantics of the system $E$ we must introduce
a set of constraints on it. A \emph{well-formed TFG $G$ built from $E$} is a
graph with one node for every variable and constant occurring in $E$, such that
we can find one set of arcs, either $X \red v$ or $v \agg X$, for every equation
of the form $v = \sum_{v_i \in X} v_i$ in $E$. We deal with inequalities by
adding slack variables. We also impose additional constraints which reflect that
the same place cannot be removed more than once. Note that the places of $N_2$
are exactly the root of $G$ (if we forget about constants).

\begin{definition}[Well-formed TFG]
  \label{def:well-formed-tfg}
  A TFG $G = (P, S, \tfgR, \tfgA)$ for the equivalence statement $(N_1, m_1)
  \eequiv (N_2, m_2)$ is \emph{well-formed} when the following constraints are
  met, with $P_1, P_2$ the set of places in $N_1, N_2$:
  \begin{description}

  \item[\textbf{(T1)}] {nodes in $S$ are roots}: if $v \in S$ then $v$ is a root
  of $G$;
  \item[\textbf{(T2)}] {nodes can be removed only once}: it is not possible to
  have $v \agg w$ and $v' \rightarrow w$ with $v \neq v'$, or to have both $v
  \red w$ and $v \agg w$;
  \item[\textbf{(T3)}] {$G$ contains all and only the equations in $E$}: we have
  $v \agg X$ or $X \red v$ if and only if the equation $v = \sum_{v_i \in X}
  v_i$ is in $E$;
  \item[\textbf{(T4)}] {$G$ is acyclic} and roots in $P \setminus S$ are exactly
  the set $P_2$.
  \end{description}
\end{definition}

Given a relation $(N_1, m_1) \eequiv (N_2, m_2)$, the well-formedness conditions
are enough to ensure the unicity of a TFG (up-to the choice of constant nodes)
when we set each equation to be either in $A$ or in $R$. In this case, we denote
the graph~$\tfg{E}$. In practice, we use the tool \reduce~\cite{tinaToolbox} to
generate the reduction system $E$.

\subsubsection*{Formula Rewriting.}
We assume given a relation $(N_1, m_1) \reduc_E (N_2, m_2)$, and its associated
well-formed TFG written $\tfg{E}$.
We consider that $F_1$ is a cube of $n$ literals, $F_1 \defeq \bigwedge_{i \in
1..n} \alpha^0_i$.
Our algorithm rewrites each $\alpha^0_i$ by applying iteratively an elimination
step, described next, according to the constraints expressed in $\tfg{E}$.
The final result is a conjunction $F_2 = \bigwedge_{i \in 1..n} \beta_i$, where
each literal $\beta_i$ has support in $N_2$.
Rewriting can only replace a variable with a group of other variables that are
its predecessors in the TFG, which ensures termination in polynomial time (in
the size of $E$).
Although the result has the same number of literals, it usually contains many
redundancies and trivial constant comparisons, so that, after simplification,
$F_2$ can actually be much smaller than $F_1$.

A reduction step (to be applied repeatedly) takes as input the current set of
literals, $\C = (\alpha_i)_{i \in 1..n}$, and modifies it. To ease the
presentation, we also keep track of a set of variables, $\B$ such that
$\bigcup_{i \in 1..n} \fv{\alpha_i} \subseteq \B$.

An elimination step is a reduction written $(\B,\C) \mapsto (\B', \C')$ where
$\C = (\alpha_i)_{i \in 1..n}$ and $\B'\subsetneq \B$, defined as one of the
three cases below (one for redundancy, and two for agglomerations, depending on
whether the removed variables are polarized or not). We assume that literals are
in normal form and that $X$ is a set of variables $\{x_1, \dots, x_k\}$.
Note the precondition $\succs{X} \cap B = \emptyset$ on all rules, which forces
them to be applied bottom-up on the TFG (remember it is a DAG).
We gave a short example of how to apply rules (AGP) and (AGD) at the end of the
previous section.\par

\begin{description}
\item[(RED)] If $X \red p$ and $\succs{p} \cap B = \emptyset$ then $(\B, \C)
\mapsto (\B' , \C')$ holds, where $\B'= \B \setminus \{p\}$ and $\C'$ is the set
of literals $\alpha'_i$ obtained by normalizing the linear constraint
$\alpha_i\{p \leftarrow x_1 + \dots + x_k\}$. That is, we substitute $p$ with
$\sum_{x_i \in X} x_i$ in $\C$, which is the meaning of the redundancy equation
(constraint (T3) in Def. \ref{def:well-formed-tfg}).\\

\item[(AGP)] If $a \agg X$ with $\succs{X} \cap B = \emptyset$, $a \in \B$, and
$X$ polarized with respect to $\C$. Then $(\B, \C) \mapsto (\B' , \C')$ holds,
where $\B'= \B \setminus X$, and, by taking $x_j \in \HLF{\C}$, we define $\C'$
as the set of literals $\alpha'_i$ obtained by normalizing the linear constraint
$\alpha_i\{x_l \leftarrow 0\}_{l \neq j}\{ x_j \leftarrow a \}$. That is, we
eliminate the variables $x_l$, different from $x_j$, from $\C$ and replace $x_j$
with $a$; where $x_j$ is a variable of $X$ that always have the highest
coefficient in each literal (among the ones of $X$).\\

\item[(AGD)] If $a \agg X$ with $\succs{X} \cap B = \emptyset$, $a \in \B$, and
$X$ is not polarized with respect to $\C$. Then $(\B, \C) \mapsto (\B' , \C')$
holds, where $\B'= \B \setminus X$ and $\C'$ is the set of literals $\alpha'_i$
obtained by normalizing the linear constraint $\alpha_i\{x_l \leftarrow 0\}_{l
\neq j}\{ x_j \leftarrow a \}$ such that $\alpha_i(x_j) = \min_X \alpha_i$.
Meaning we eliminate the variables $x_l$ different from $x_j$ from $\alpha_i$
and replace $x_j$ with $a$, where $x_j$ is a variable with the smallest
coefficient in $\alpha_i$ (among the ones of $X$). Note that the chosen variable
$x_j$ is not necessarily the same in every literal of $\C$.\par
\end{description}

Our goal is to preserve the semantics of formulas at each reduction step, in the
sense of the relations $\eimplies$ and $\eequiv$. In the following, we use $\C$
to represent both a set of literals $(\alpha_i)_{i \in I}$ and the cube formula
$\bigwedge_{i \in I} \alpha_i$. We can prove that the elimination steps
corresponding to the redundancy (RED) and polarized agglomeration (AGP) cases
preserve the semantics of the formula $\C$. On the other hand, a non-polarized
agglomeration step (AGD) may lose some markings.

\subsubsection*{Proof of the Algorithm.}

We prove the main result of the paper, namely that fast quantifier elimination
preserves the integer solutions of a system when we only have polarized
agglomerations. To this end, we need to prove two theorems. First,
Theorem~\ref{th:proj}, which entails the soundness of one step of elimination.
It also entails completeness for rules (RED) and (AGP). Second, we prove a
progress property (Th.~\ref{th:progress} below), which guarantees that we can
apply elimination steps until we reach a set of literals $C'$ with support on
the reduced net $N_2$.

\begin{theorem}[Projection Equivalence]\label{th:proj} If $(\B, \C) \mapsto
 (\B', \C')$ is a (RED) or (AGP) reduction then $\C' \eequiv \C$; otherwise $\C'
 \eimplies \C$.
\end{theorem}

We prove Th.~\ref{th:proj} in two steps. We start by proving that elimination
steps are sound, meaning that the integer solutions of $\C'$ are also solutions
of $\C$ (up-to $E$). Then we prove that elimination is also complete for rules
(RED) and (AGP). In the following, we use $\C$ to represent both a set of
literals $(\alpha_i)_{i \in I}$ and the cube formula $\bigwedge_{i \in I}
\alpha_i$.

\begin{lemma}[Soundness]
  If $(\B, \C) \mapsto (\B' , \C')$ then $\C' \eimplies \C$.
\end{lemma}

\begin{proof}
  Take a valuation $m'$ of $\mathbb{N}^{\B'}$ such that $m' \models E \land
  \C'$. We want to show that there exists a marking $m$ of $\mathbb{N}^\B$ such
  that $m \eequiv m'$ satisfying $E \land \C$. 
  
  We have three possible cases, corresponding to rule (RED), (AGP) or (AGD). In
  each case, we provide a marking $m$ built from $m'$. Since $m \eequiv m'$ is
  enough to prove $m \models E$, we only need to check two properties: first
  that $m \eequiv m'$ ($\ast$), then that $m \models \alpha$ for every literal
  $\alpha$ in $\C$ ($\ast\ast$). 
  
  \begin{description}
    \item[(RED)] In this case we have $X \red p$ and $\B' = \B \setminus \{p\}$,
    with $X = \{x_1, \dots, x_k\}$. We can extend $m'$ into the unique valuation
    $m$ of $\mathbb{N}^\B$ such that $m(p) = m'(x_1) + \dots + m'(x_k)$ and
    $m(v) = m'(v)$ for all other nodes $v$ in $\B \setminus \{p\}$. Since $p =
    x_1 + \dots + x_k$ is an equation of $E$ (condition (T3)) we obtain that
    $m'\eequiv m$ and therefore also $m\models E$ ($\ast$). 
    
    We now prove that $m \models \C$. The literals in $\C'$ are of the form
    $\alpha\sigma$ with $\sigma$ the substitution $\{p \leftarrow x_1 + \dots +
    x_k\}$ and $\alpha$ in $\C$. Remember that, with our notations (e.g.
    Equation~\eqref{eq:1} in page~\pageref{eq:1}), we have $m \models \alpha$ if
    and only if $\alpha\{m\}$ SAT (is satisfiable). By hypothesis, $m'\models
    \alpha\sigma$. Hence, $\alpha\sigma\{m'\}$ SAT, which is equivalent to
    $\alpha\{m\}$ SAT, and therefore $m \models \alpha$ ($\ast\ast$), as
    required.
    \medbreak
  
    \item[(AGP)] In this case we have $a \agg X$ with $X = \{x_1, \dots, x_k\}$,
    polarized relative to $\C$, and $B' = B \setminus X$. We consider $x_j$ in
    $X$ the variable in $\HLF{\C}$ that was chosen in the reduction; meaning
    that $\C'$ is a conjunction of literals of the form $\alpha\{x_l \leftarrow
    0\}_{l \neq j}\{ x_j \leftarrow a \}$, with $\alpha$ a literal of $\C$.
    Given $m'$ a model of $\C'$, we define $m$ the unique marking on $\Nat^{\B}$
    such that $m(x_j) = m(a)$, $m(x_l) = 0$ for all $l \neq j$, and $m(v) =
    m'(v)$ for all other variables $v$ in $B \setminus X$.

    From Lemma~2 of~\cite{spin2021} (the ``token propagation'' property of
    TFGs), we know that any distribution of $m(a)$ tokens, in place $a$, over
    the $(x_i)_{i \in 1..k}$, is also a model of $E$. Which means that $m
    \models E$ $(\ast)$.
    Note that the token propagation Lemma does not imply that the value of
    $m(v)$, for the nodes ``below $X$'' ($v$ in $\succs{X}$), is unchanged. This
    is not problematic, since the side condition $\succs{X} \cap \B = \emptyset$
    ensures that these nodes are not in $\B$, and therefore cannot influence the
    value of $\alpha\{ m \}$. 
    
    Consider a literal $\alpha$ in $\C$. Since $m'\models \C'$, we have that
    $\alpha\{x_l \leftarrow 0\}_{l \neq j}\{ x_j \leftarrow a \}\{ m'\}$ SAT,
    which is exactly $\alpha\{ m \}$, since $\succs{X} \cap B = \emptyset$, as
    needed ($\ast\ast$).
    \medbreak

    \item[(AGD)]  In this case we have $a \agg X$ with $X = \{x_1, \dots,
    x_k\}$, non-polarized relative to $\C$, and $B' = B \setminus X$. We know
    that $m' \models E$, therefore there is a marking $m$ of $\Nat^{\B}$ that
    extends $m'$ such that $m \eequiv m'$ ($\ast$).
    
    Consider a literal $\alpha$ in $\C$. By definition of (AGD), we have an
    associated literal $\alpha'\defeq \alpha\{x_l \leftarrow 0\}_{l \neq j}\{
    x_j \leftarrow a \}$ in $\C'$ such that $\alpha(x_j) = \min_X \alpha_i$.
    Since the coefficient of $x_j$ is minimal, we have that $\sum_{i \in 1..k}
    \alpha(x_i)\, m(x_i) \geq \alpha(x_j)\, \sum_{i \in 1..k} m(x_i) =
    \alpha(x_j) \, m'(a)$, and therefore $\sum_{v \in \B} \alpha(v) \, m(v) \geq
    \sum_{v \in \B'} \alpha'(v) \, m'(v)$. The result follows from the fact that
    $\alpha\{ m\}$ SAT ($\ast\ast$).
  \end{description}
  \qed
\end{proof} 

Now we prove that our quantifier elimination step, for the (RED) and (AGP)
cases, leads to a complete projection, that is any solution of the initial
formula corresponds to a projected solution in the projected formula.

\begin{lemma}[Completeness]
  If $(\B, \C) \mapsto (\B', \C')$ is a (RED) or (AGP) reduction then $\C
  \eimplies \C'$.
\end{lemma}

\begin{proof}
  Take a marking $m$ of $\mathbb{N}^\B$ such that $m \models E \land \C$. We
  want to show that there exists a valuation $m'$ of $\mathbb{N}^{\B'}$ such
  that $m \eequiv m'$ ($\ast$) and $m' \models \C'$ ($\ast\ast$). This is enough
  to prove $m' \models E \land C'$. We have two different cases, corresponding
  to the rules (RED) and (AGP).

\begin{description}
  \item[(RED)] In this case we have $X \red p$ with $X = \{ x_1 , \dots,x_k \}$
  and $B' = B \setminus \{p\}$. We define $m'$ as the (unique) projection of $m$
  on $\B'$. Since $m \models E$ we have that $m'\eequiv m$ ($\ast$).
  
  Also, literals in $\C'$ are of the form $\alpha'\defeq \alpha\{ p \leftarrow
  x_1 + \dots + x_k \}$ where $\alpha$ is a literal of $\C$. Since $m(p) =
  \sum_{i \in 1..k} m(x_i)$ and $m$ is a model of $\alpha$, it is also the case
  that $m'$ is a model of $\alpha'$ $(\ast\ast)$.
  \medbreak

  \item[(AGP)] In this case we have $a \agg X$ with $X = \{x_1, \dots, x_k\}$
  and $\B' = B \setminus X$. We define $m'$ as the (unique) projection of $m$ on
  $\B'$, by taking $m'(a) = \sum_{i \in 1..k} m(x_i)$. Since $m \models E$ we
  have that $m'\eequiv m$ ($\ast$).

  We consider $x_j$ in $X$ the variable in $\HLF{\C}$ that was chosen in the
  reduction; meaning that $\C'$ is a conjunction of literals of the form
  $\alpha\{x_l \leftarrow 0\}_{l \neq j}\{ x_j \leftarrow a \}$, with $\alpha$ a
  literal of $\C$. Since $\sum_{i \in 1..k} \alpha(x_i) \, m(x_i) \leqslant
  \alpha(x_j) \, \sum_{i \in 1..k} m(x_i) = \alpha(x_j) \, m'(a)$, we have $m'$
  is a model of $\alpha'$ ($\ast\ast$).\qed
  \end{description}
\end{proof}

The final step of our proof relies on a \emph{progress property}, meaning there
is always a reduction step to apply except when all the literals have their
support on the reduced net, $N_2$. This property relies on relation $\mapsto^*$,
which is the transitive closure of $\mapsto$. Together with Th.~\ref{th:proj},
the progress theorem ensures the existence of a sequence $(P, \C) \mapsto^*
(P_2, \C')$, such that $C \eequiv C'$ (or $C' \eimplies C$ if we have at least
one non-polarized agglomeration). In this context, $P$ is the set of all
variables occurring in the TFG of $E$, and therefore it contains $P_1 \cup P_2$.

\begin{theorem}[Progress]
  \label{th:progress}
  Assume $(P, F_1) \mapsto^* (\B, \C)$ then either $\B \subseteq P_2$, the set
  of places of $N_2$, or there is an elimination step $(\B, \C) \mapsto (\B',
  \C')$ such that $\fv{\C'} \subseteq \B'$ and the places removed from $B$ have
  no successors in $\B'$: for all places $p$ in $\B \setminus \B'$, we have
  $\succs{p} \cap \B = \emptyset$. 
\end{theorem}

\begin{proof}
  Assume we have $(P, F_1) \mapsto^* (\B, \C)$ and $\B \nsubseteq P_2$.
  
  By condition (T4) in Def.~\ref{def:well-formed-tfg}, we know that $P_2$ are
  roots in the TFG $\tfg{E}$. We consider the set of nodes in $B \setminus P_2$,
  which corresponds to nodes in $B$ with at least one parent.
  Also, by condition (T4), we know that $\tfg{E}$ is acyclic, then there are nodes
  in $\B \setminus P_2$ that have no successors in $\B$. We call this set $L$.
  Hence, $L \defeq \{v \mid v \in \B \setminus P_2 \land \succs{v} \cap B =
  \emptyset\}$. 
  
  Take a node $p$ in $L$. We have two possible cases. If there is a set $X$ such
  that $X \red p$, we can apply the (RED) elimination rule.
  Otherwise, there exists a node $a$ and a set $X \subseteq L$ (by condition (T2))
  such that $a \agg X$ with $p \in X$. In this case, apply rule (AGP) or (AGD),
  depending on whether the agglomeration is polarized or not.\qed
  \end{proof}

\medskip\noindent\textbf{Remark.} We have designed the rule (AGD) to obtain at
least $F_2 \eimplies F_1$ when the procedure is not complete (instead of $F_2
\eequiv F_1$), which is useful for finding witnesses (see
Th.~\ref{th:exploration}). Alternatively, we could propose a variant rule, say
(AGD'), which chooses the variable $x_j$ having the highest coefficient in
$\alpha_i$, that is $\alpha_i(x_j) = \max_{X}\alpha_i$. This variant guarantees
a dual result, that is $F_1 \eimplies F_2$. In this case, if $F_2$ is not
reachable then $F_1$ is not reachable, which is useful to prove invariants.


\section{Experimental Results}
\label{sec:performance}

\paragraph*{Data-Availability Statement.} We have implemented our fast
quantifier elimination procedure in a new, open-source tool called
\tool{Octant}~\cite{octant}, available on GitHub. All the tools, scripts and
benchmarks used in our experiments are part of our
artifact~\cite{amat_2023_7935154}.

We use an extensive, and independently managed, set of models and formulas
collected from the 2022 edition of the Model Checking Contest
(MCC)~\cite{mcc:2022}.
The benchmark is built from a collection of $128$ models. Most models are
parametrized and can have several instances. This amounts to about $1\, 400$
different instances of Petri nets whose size varies widely, from $9$ to $50\,
000$ places, and from $7$ to $200\, 000$ transitions. This collection provides a
large number of examples with various structural and behavioral characteristics,
covering a large variety of use cases. Each year, the MCC organizers randomly
generate $16$ reachability formulas for each instance. The pair of a Petri
instance and a formula is a \emph{query}; which means we have more than $22\,
000$ queries overall.

We do not compute reductions ourselves but rely on the tool \reduce, part of the
latest public release of \tina~\cite{tinaToolbox}. We define the reduction ratio
($r_p$) of an instance as the ratio $(p_{init} - p_{red})/p_{init}$ between the
number of places before ($p_{init}$) and after ($p_{red}$) reduction. We only
consider instances with a ratio, $r_p$, between $1\%$ and $100\%$ (meaning the
net is \emph{fully reduced}), which still leaves about $17\,000$ queries, so
about $77\%$ of the whole benchmark. More information about the distribution of
reductions can be found in~\cite{fi2022,berthomieu_counting_2019}, where we show
that almost half the instances ($48\%$) can be reduced by a factor of $25\%$ or
more. 

The size of the reduction system, $E$, is proportional to the number of places
that can be removed. To give a rough idea, the mean number of variables in $E$
is $1\,375$, with a median value of $114$ and a maximum of about $62\,000$. The
number of literals is also rather substantial: a mean of $869$ literals ($62\%$
of agglomerations and $38\%$ of redundancies), with a median of $27$ and a
maximum of about $38\,000$.

We report on the results obtained on two main categories of experiments: first
with model-checking, to evaluate if our approach is effective in practice, using
real tools; then to assess the precision and performance of our fast quantifier
elimination procedure.

\begin{figure}[t] 
  \label{fig8} 
  \begin{minipage}[b]{0.5\linewidth}
    \centering
    \includegraphics[width=\linewidth]{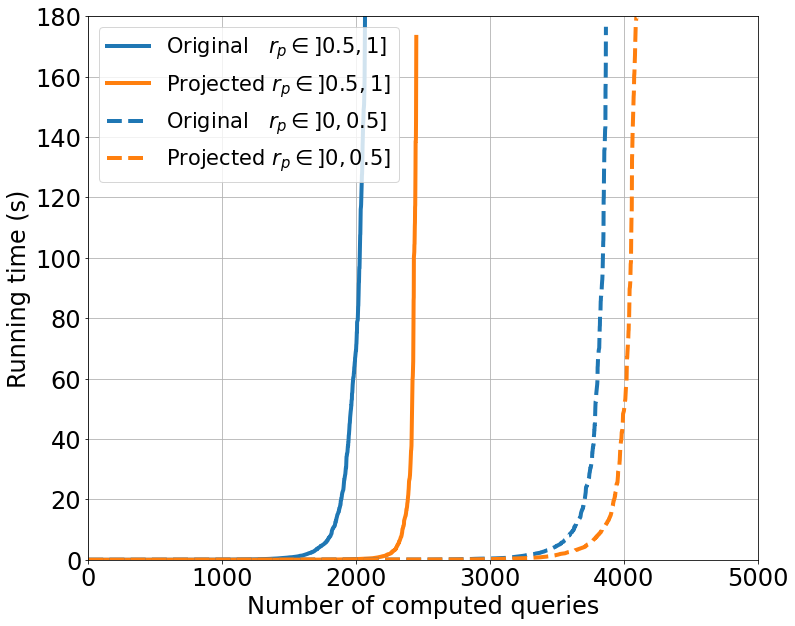} 
    \caption{Random walk w/wo reductions.} 
    \label{fig:walk}
    \vspace{4ex}
  \end{minipage}
  \begin{minipage}[b]{0.5\linewidth}
    \centering
    \includegraphics[width=\linewidth]{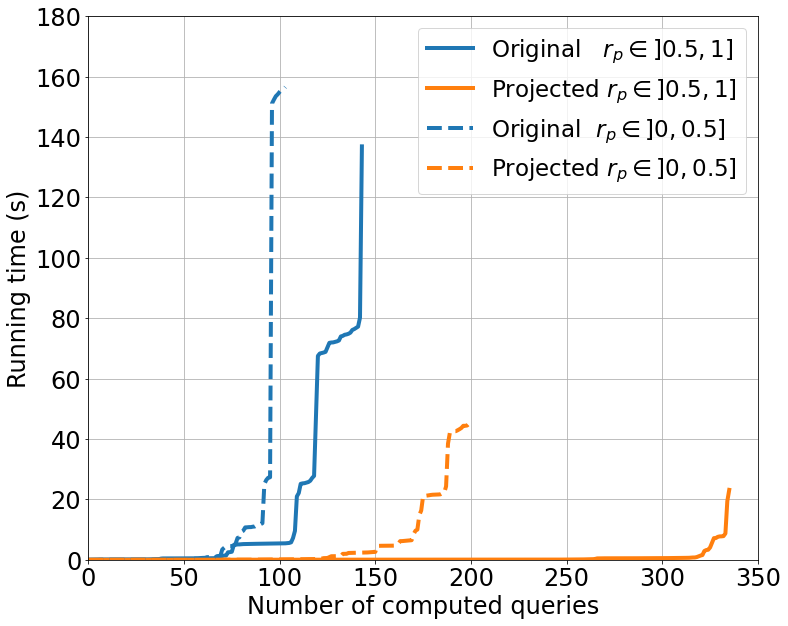} 
    \caption{$k$-induction w/wo reductions.} 
    \label{fig:kinduc}
    \vspace{4ex}
  \end{minipage} 
\end{figure}

\subsubsection*{Model-Checking.}
We start by showing the effectiveness of our approach on both $k$-induction and
random walk. This is achieved by comparing the computation time, with and
without reductions, on a model-checker that provides a ``reference''
implementation of these techniques. (Without any other optimizations that could
interfere with our experiments.) It is interesting to test the results of our
optimization separately on these two techniques. Indeed, each technique is
adapted to a different category of queries: properties that can be decided by
finding a witness, meaning true EF formulas or false AG ones, can often be
checked more efficiently using a random exploration of the state space. On the
other hand, symbolic verification methods are required when we want to check
invariants.

We display our results using the two ``cactus plots'' in Figs.~\ref{fig:walk}
and~\ref{fig:kinduc}. We distinguish between two categories of instances,
depending on their reduction ratio. We use dashed lines for models with a low or
moderate reduction ratio (value of $r_p$ less than $50\%$) and solid lines for
models that can be reduced by more than half. The first category amounts to
roughly $10\,700$ queries ($62\%$ of our benchmark), while the second category
contains about $6\,000$ queries. The most interesting conclusion we can draw
from these results is the fact that our approach is beneficial even when there
is only a limited amount of reductions.

Our experiments were performed with a maximal timeout of $\SI{180}{\second}$ and
integrated the projection time into the total execution time. The ``vertical
asymptote'' that we observe in these plots is a good indication that we cannot
expect further gains, without projection, for timeout values above 60s. Hence,
our choice of timeout value does not have an overriding effect. We observe
moderate performance gains with random exploration (with $\times 1.06$ more
computed queries on low-reduction instances and $\times 1.19$ otherwise) and
good results with $k$-induction (respectively $\times 1.94$ and $\times 2.33$).

We obtain better results if we focus on queries that take more than
$\SI{1}{\second}$ on the original formula, which indicates that reductions are
most effective on ``difficult problems'' (there is not much to gain on instances
that are already easy to solve). With random walk, for instance, the gain
becomes $\times 1.48$ for low-reduction instances and $\times 1.93$ otherwise. 
The same observation is true with $k$-induction, with performance gains of $\times 3.78$ and $\times 3.51$ respectively.


\subsubsection*{Model-Checking under Real Conditions.}
We also tested our approach by transparently adding polyhedral reductions as a
front-end to three different model-checkers: \tapaal~\cite{tapaal},
\its~\cite{its_tools}, and \lola~\cite{lola}, that implement portfolios of
verification techniques.
All three tools regularly compete in the MCC (on the same set of queries that we
use for our benchmark), \tapaal{} and \its{} share the top two places in the
reachability category of the 2022 and 2023 editions.

We ran each tool on our set of complete projections, which amounts to almost
$100\,000$ runs (one run for each tool, once on both the original and the
projected query). We obtained a $100\%$ reliability result, meaning that all
tools gave compatible results on all the queries; and therefore also compatible
results on the original and the projected formulas.

A large part of the queries can be computed by all the tools in less than
$\SI{100}{\milli\second}$ and can be considered as {easy}. These queries are
useful for testing reliability but can skew the interpretation of results when
comparing performances. This is why we decided to focus our results on a set of
$809$ \emph{challenging queries}, that we define as queries for which either
\tapaal{} or \its{}, or both, are not able to compute a result before
projection. The $809$ challenging queries ($4\%$ of queries) are well
distributed, since they cover $209$ different instances ($14\%$ of all
instances), themselves covering $43$ different models ($33\%$ of the models).

We display the results obtained on the challenging queries, for a timeout of
$\SI{180}{\second}$, in Table~\ref{fig:table}. We provide the number of computed
queries before and after projection, together with the mean and median speed-up
(the ratio between the computation time with and without projection). The
``Exclusive'' column reports, for each tool, the number of queries that can only
be computed using the projected formula. Note that we may sometimes time out
with the projected query, but obtain a result without. This can be explained by
cases where the size of the formula blows up during the transformation into DNF.

\newcolumntype{x}[1]{>{\centering\arraybackslash\hspace{0pt}}p{#1}}
\begin{figure}[t]
  \begin{center}
    \begin {tabular}{l cc x{0em} rc x{0em} c}%
      \toprule
      \textsc{Tool} & \multicolumn{2}{c}{\textsc{\# Queries}} 
      & & \multicolumn{2}{c}{\textsc{Speed-up}} 
      & & \multirow{2}{*}{\begin{minipage}[c]{4em}\centering 
        \textsc{\# Excl.\\ Queries} 
      \end{minipage}}\\\cmidrule(rl){2-3}\cmidrule(rl){5-6}
       & \textsc{Original} & \textsc{Projected} 
       && \textsc{Mean} & \textsc{Median} & & \\\midrule
      \its &  $302$ & $352$ && $ 1.42$ & $ 1.00$ && $78$\\
      \lola &  $143$ & $205$ && $ 14.97$ & $ 1.44$ && $76$\\
      \tapaal &  $134$ & $216$ && $ 1.87$ & $ 1.17$ && $99$\\
      \bottomrule
    \end{tabular}
  \end{center}
\caption{Impact of projections on the challenging queries.\label{fig:table}}
\end{figure}

We observe substantial performance gains with our approach and can solve about
half of the challenging queries. For instance, we are able to compute $\times
1.6$ more challenging queries with \tapaal\ using projections than without. (We
display more precise results on \tapaal, the winner of the MCC 2022 edition, in
Fig.~\ref{fig:tapaal}.) We were also able to compute $62$ queries, using
projections, that no tool was able to solve during the last MCC (where each tool
has $\SI{3\,600}{\second}$ to answer $16$ queries on a given model instance).
All these results show that polyhedral reductions are effective on a large set
of queries and that their benefits do not significantly overlap with other
existing optimizations, an observation that was already made, independently,
in~\cite{fi2022} and~\cite{bonneland_stubborn_2019}.

The approach implemented in \tool{Octant}{} was partially included in the
version of our model-checker, called \tool{SMPT}~\cite{amat2023smpt}, that
participated in the MCC 2023 edition. We mainly left aside the handling of
under-approximated queries, when the formula projection is not complete. While
\tool{SMPT}{} placed third in the reachability category, the proportion of
queries it was able to solve raised by $5.5\%$ between 2022 (without the use of
\tool{Octant}{}) and 2023, to reach a ratio of $93.6\%$ of all queries solved
with our tool. A $5\%$ gain is a substantial result, taking into account that
the best-performing tools are within $1\%$ of each other; the ratios for \its{}
and \tapaal{} in 2023 are respectively $94.6\%$ and $94.3\%$.

\subsubsection*{Performance of Fast Elimination.}
Our last set of experiments is concerned with the accuracy and performance of
our quantifier elimination procedure. We decided to compare our approach with
\tool{Redlog}~\cite{dolzmann1997redlog} and \tool{isl}~\cite{verdoolaege2010isl}
(we give more details on these two tools in
Sect.~\ref{sec:discussion-related-works}).

\begin{figure}[t] 
  \label{fig7} 
  \begin{minipage}[b]{0.5\linewidth}
    \centering
    \includegraphics[width=\linewidth]{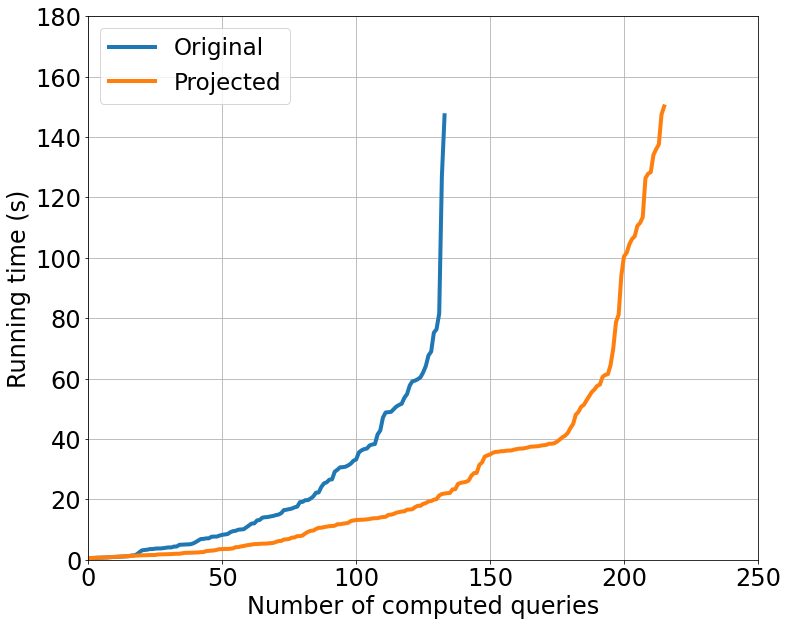} 
    \caption{\tapaal{} w/wo polyhedral reductions.}
    \label{fig:tapaal}
    \vspace{4ex}
  \end{minipage}
  \begin{minipage}[b]{0.5\linewidth}
    \centering
    \includegraphics[width=\linewidth]{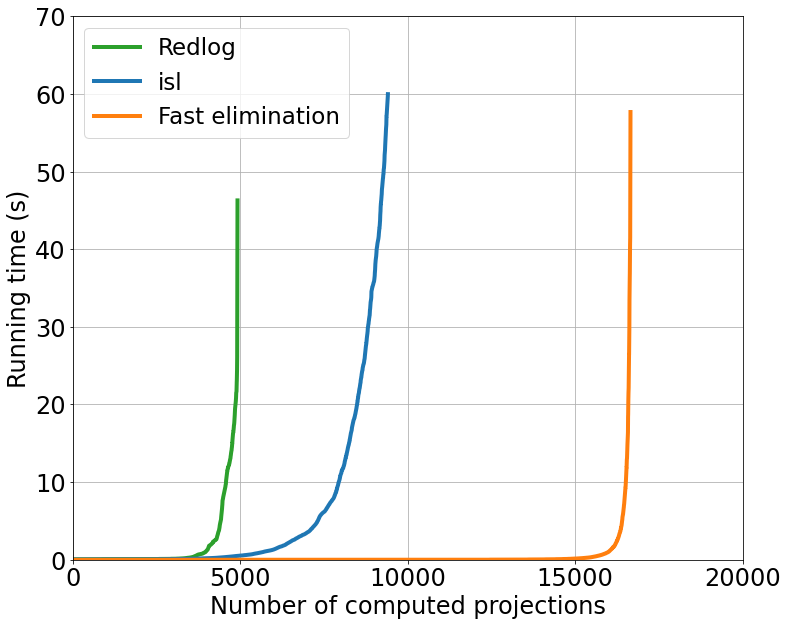} 
    \caption{\tool{Redlog} vs \tool{isl} vs Fast elimination.}
    \label{fig:barvinok}
    \vspace{4ex}
  \end{minipage} 
\end{figure}

We display our results in the cactus plot of Fig.~\ref{fig:barvinok}, where we
compare the number of projections we can compute given a fixed timeout.
We observe a significant performance gap. For instance, with a timeout of
$\SI{60}{\second}$, we are able to compute $16\,653$ projections, out of
$16\,976$ queries ($98\%$), compared to $9\,407$ ($55\%$) with \tool{isl} or
$4\,915$ ($29\%$) with \tool{Redlog}. So an increase of $\times 1.77$ over the
better of the two tools.
This provides more empirical evidence that the class of linear systems we manage
is not trivial, or at least does not correspond to an easy case for the
classical procedure implemented in \tool{Redlog} and \tool{isl}.
We also have good results concerning the precision of our approach, since we
observe that about $80\%$ of the projections are complete. Furthermore,
projections are inexpensive. For instance, the computation time is less than
$\SI{1}{\second}$ for $96\%$ of the formulas. We also obtained a median
reduction ratio (computed as for the number of places) of $0.2$ for the number
of cubes and their respective number of literals.


\section{Discussion and Related Works}
\label{sec:discussion-related-works}

We proposed a quantifier elimination procedure that can take benefit of
polyhedral reductions and can be used, transparently, as a pre-processing step
of existing model-checkers. The main characteristic of our approach is to rely
on a graph structure, called Token Flow Graph, that encodes the specific shape
of our reduction equations.

The idea of using linear equations to keep track of the effects of structural
reductions originates from~\cite{berthomieu2018petri,berthomieu_counting_2019},
as a method for counting the number of reachable markings. We extended this
approach to Bounded Model Checking (BMC) in~\cite{fi2022} where we defined our
\emph{polyhedral abstraction} equivalence, $\eequiv$, and in~\cite{reductron} where we
proposed a procedure to automatically prove when such abstractions are correct.
The idea to extend this relation to formulas is new (see
Def.~\ref{def:feequiv}). In this paper, we broaden our approach to a larger set
of verification methods; most particularly $k$-induction, which is useful to
prove invariants, and simulation (or \emph{random walk}), which is useful for
finding counter-examples. We introduced the notion of a Token Flow Graph (TFG)
in~\cite{spin2021}, as a new method to compute the \emph{concurrent places} of a
net; meaning all pairs of places that can be marked simultaneously. We find a
new use for TFGs here, as the backbone of our variable elimination algorithm,
and show that we can efficiently eliminate variables in systems of the form $E
\wedge F$, for an arbitrary $F$.

We formulated our method as a variable elimination procedure for a restricted
class of linear systems. There exist some well-known classes of linear systems
where variable elimination has a low complexity. A celebrated example is given
by the link between unimodular matrices and integral
polyhedra~\cite{hoffman2010integral}, which is related to many examples found in
abstract domains used in program verification, such as \emph{systems of
differences}~\cite{aspvall1980polynomial} or
octagon~\cite{mine2006octagon,jeannet2009apron}. To the best of our knowledge,
none of the known classes correspond to what we define using TFGs. There is also
a rich literature about quantifier elimination in Presburger arithmetic, such as
Cooper's algorithm~\cite{cooper1972theorem,haase2018survival} or the Omega
test~\cite{omega} for instance, and how to implement it
efficiently~\cite{huynh1992practical,lasaruk2007weak,monniaux2010quantifier}.
These algorithms have been implemented in several tools, using many different
approaches: automata-based, e.g.~\tool{TaPAS}~\cite{leroux2009tapas}; inside
computer algebra systems, like with \tool{Redlog}~\cite{dolzmann1997redlog}; or
in program analysis tools, like \tool{isl}~\cite{verdoolaege2010isl}, part of
the Barvinok toolbox. Another solution would have been to retrieve “projected
formulas” directly from SMT solvers for linear arithmetic, which often use
quantifier elimination internally. Unfortunately, this feature is not available,
even though some partial solutions have been proposed
recently~\cite{barth2022ultimate}. All the exact methods that we tested have
proved impractical in our case. This was to be expected. Quantifier elimination
can be very complex, with an exponential time complexity in the worst case (for
existential formulas as we target); it can generate very large formulas; and it
is highly sensitive to the number of variables, when our problem often involves
several hundreds and sometimes thousands of variables. Also, quantifier
elimination often requires the use of a divisibility operator (also called
\emph{stride format} in~\cite{omega}), which is not part of the logic fragment
that we target.

Another set of related works is concerned with \emph{polyhedral
techniques}~\cite{feautrier2011polyhedron}, used in program analysis. For
instance, our approach obviously shares similarities with works that try to
derive linear equalities between variables of a
program~\cite{cousot1978automatic}, and polyhedral abstractions are very close
in spirit to the notion of linear dependence between vectors of integers
(markings in our case) computed in compiler optimizations. Another indication of
this close relationship is the fact that \tool{isl}, the numerical library that
we use to compare our performances, was developed to support polyhedral
compilation. We need to investigate this relationship further and see if our
approach could find an application with program verification. From a more
theoretical viewpoint, we are also looking into ways to improve the precision of
our projection in the cases where we find non-polarized sets of constraints.

\clearpage
\newpage
\bibliographystyle{alphaurl}
\bibliography{bibfile}

\end{document}